\documentclass[aps,pra,twocolumn,nofootinbib,superscriptaddress]{revtex4-2}

\usepackage[]{graphics,graphicx,epsfig}
\usepackage{subfigure}
\usepackage{bbm}
\usepackage{amsmath}
\usepackage{amssymb}
\usepackage{amsfonts}
\usepackage{mathrsfs}
\usepackage{bm}
\usepackage{url}
\usepackage[T1]{fontenc}
\usepackage{csquotes}
\MakeOuterQuote{"}

\usepackage{algorithm}
\usepackage{algorithmicx}
\usepackage{algpseudocode}

\usepackage{dcolumn}
\usepackage{color}
\usepackage{hyperref}
\usepackage{booktabs}
\usepackage{nicefrac} 
\usepackage{amsthm}
\usepackage{mathdots}
\usepackage{amsfonts,amssymb,amsmath}
\usepackage[dvipsnames]{xcolor}

\newcommand{\caH}{\mathcal{H}}

\newcommand{\tr}{\operatorname{tr}}

\newtheorem{lemma}{Lemma}
\newtheorem{proposition}{Proposition}

\def\<{\langle}  
\def\>{\rangle}  
\def\eqref#1{\textup{(\ref{#1})}}
\newcommand{\eref}[1]{Eq.~\textup{(\ref{#1})}}
\newcommand{\lref}[1]{Lemma~\ref{#1}}

\usepackage{enumitem}      
\renewcommand{\figurename}{Figure}
\begin{document}
	\preprint{}
	
	\title{Experimental Verification of Entangled States in the Adversarial Scenario}
	
\author{Wen-Hao Zhang}
\thanks{These authors contributed equally to this work}
\affiliation{CAS Key Laboratory of Quantum Information, University of Science and Technology of China, Hefei, Anhui 230026, China}
\affiliation{CAS Center For Excellence in Quantum Information and Quantum Physics, University of Science and Technology of China, Hefei, Anhui 230026, China}
\affiliation{School of Physics and Optoelectronics Engineering, Anhui University, Hefei, Anhui 230601, China}

\author{Zihao Li}
\thanks{These authors contributed equally to this work}
\affiliation{State Key Laboratory of Surface Physics, Department of Physics, and Center for Field Theory and Particle Physics, Fudan University, Shanghai 200433, China}
\affiliation{Institute for Nanoelectronic Devices and Quantum Computing, Fudan University, Shanghai 200433, China}
\affiliation{Shanghai Research Center for Quantum Sciences, Shanghai 201315, China}

\author{Gong-Chu Li}
\affiliation{CAS Key Laboratory of Quantum Information, University of Science and Technology of China, Hefei, Anhui 230026, China}
\affiliation{CAS Center For Excellence in Quantum Information and Quantum Physics, University of Science and Technology of China, Hefei, Anhui 230026, China}
\affiliation{Hefei National Laboratory,  Hefei 230088, China}

\author{Xu-Song Hong}
\affiliation{CAS Key Laboratory of Quantum Information, University of Science and Technology of China, Hefei, Anhui 230026, China}
\affiliation{CAS Center For Excellence in Quantum Information and Quantum Physics, University of Science and Technology of China, Hefei, Anhui 230026, China}

\author{Huangjun~Zhu}
\email{zhuhuangjun@fudan.edu.cn}
\affiliation{State Key Laboratory of Surface Physics, Department of Physics, and Center for Field Theory and Particle Physics, Fudan University, Shanghai 200433, China}
\affiliation{Institute for Nanoelectronic Devices and Quantum Computing, Fudan University, Shanghai 200433, China}
\affiliation{Shanghai Research Center for Quantum Sciences, Shanghai 201315, China}
\affiliation{Hefei National Laboratory,  Hefei 230088, China}
	
\author{Geng Chen}
	\email{chengeng@ustc.edu.cn}
	\affiliation{CAS Key Laboratory of Quantum Information, University of Science and Technology of China, Hefei, Anhui 230026, China}
	\affiliation{CAS Center For Excellence in Quantum Information and Quantum Physics, University of Science and Technology of China, Hefei, Anhui 230026, China}
    \affiliation{Hefei National Laboratory,  Hefei 230088, China}

	\author{Chuan-Feng Li}
	\email{cfli@ustc.edu.cn}
	\affiliation{CAS Key Laboratory of Quantum Information, University of Science and Technology of China, Hefei, Anhui 230026, China}
	\affiliation{CAS Center For Excellence in Quantum Information and Quantum Physics, University of Science and Technology of China, Hefei, Anhui 230026, China}
    \affiliation{Hefei National Laboratory,  Hefei 230088, China}
	
	\author{Guang-Can Guo}
	\affiliation{CAS Key Laboratory of Quantum Information, University of Science and Technology of China, Hefei, Anhui 230026, China}
	\affiliation{CAS Center For Excellence in Quantum Information and Quantum Physics, University of Science and Technology of China, Hefei, Anhui 230026, China}
    \affiliation{Hefei National Laboratory,  Hefei 230088, China}
	
	\date{\today}
	
\begin{abstract}
Efficient verification of entangled states is crucial to many applications in quantum information processing. However, the effectiveness of standard quantum state verification (QSV) is based on the condition of independent and identical distribution (IID), which impedes its applications in many practical scenarios. Here we demonstrate a defensive QSV protocol, which is effective in all kinds of  non-IID scenarios, including the extremely challenging adversarial scenario. To this end, we build a high-speed preparation-and-measurement apparatus controlled by quantum random-number generators. Our experiments clearly show that standard QSV protocols 
often provide unreliable fidelity certificates  in non-IID scenarios. In sharp contrast, the defensive QSV protocol based on a homogeneous strategy can provide  reliable and nearly tight fidelity certificates  at comparable high efficiency, even under malicious attacks. Moreover, our scheme is robust against the imperfections in a realistic experiment, which is very appealing to practical applications.
\end{abstract}

	\maketitle
	\clearpage

\section{Introduction}
Entangled quantum states play crucial roles both in foundational studies and in quantum information processing, such as quantum teleportation \cite{bennett1993teleporting,bouwmeester1997experimental},  quantum cryptography \cite{gisin2002quantum,portmann2022security}, and quantum computation \cite{raussendorf2001one,preskill2018quantum}. Hence, efficient and reliable verification of desired quantum states is an important step in practical applications of quantum technologies \cite{eisert2020quantum,kliesch2021ReviewQSV,carrasco2021theoretical,morris2022quantum,yu2022ReviewQSV,govcanin2022sample}.
However, traditional tomographic approaches are too resource-consuming to achieve this goal because  they extract too much unnecessary information \cite{resch2005full,haffner2005scalable,lvovsky2009continuous,sugiyama2013precision}. Even with popular alternative approaches, such as direct fidelity estimation \cite{flammia2011direct,PhysRevLett.107.210404,zhang2021direct}, the sample efficiency is still not satisfactory except for some special cases.

In the past few years, an alternative approach called quantum state verification (QSV) has received increasing attention because of its potential to extract the key
information---fidelity with the target state---with remarkably lower sample cost than traditional approaches \cite{pallister2018PRL,zhu2019advPRA}.
So far, efficient verification protocols based on local operations and classical communication have been constructed for various classes of quantum states
\cite{hayashi2006study,GluzKEA18,pallister2018PRL,zhu2019MES,wang2019twoqubit,li2019bipartite,yu2019bipartite,zhu2019hypergraph,liu2019Dicke,dangniam2020stabilizer,li2020GHZ,zhu2019advPRA,li2021Dicke,liu2021CVstate,wu2021CVstatePRL,chen2023AKLT,zhu2022Hamiltonian}, and the efficiency of QSV has been demonstrated in a number of experiments \cite{zhang2020PRLexperiQSV,zhang2020npjExperiQSV,jiang2020npjExperiQSV,xia2022experimental}. 

However, standard QSV (SQSV) protocols rely on the  assumption of independent and identical distribution (IID), which is not guaranteed in many practical situations, because  realistic quantum devices may suffer from correlated noise, such as unwanted interqubit coupling or quantum crosstalk \cite{Google,Crosstalk}. 
In this case, the naive acceptance of the IID assumption without rigorous correlation testing may undermine verification reliability and 
lead to incorrect predictions \cite{PhysRevA.75.052318}.
Moreover, in the adversarial scenario, which is the most challenging situation within the non-IID landscape,  the state-preparation device may even be controlled by  a potentially malicious adversary, Bob, who may try every means to cheat the client, Alice \cite{zhu2019advLett,hayashi2015PRLverifiable,zhu2019advPRA}.  Efficient and reliable QSV in this scenario is crucial to many important applications that entail high-security requirements, such as blind measurement-based quantum computation \cite{hayashi2015PRLverifiable,morimae2013blind,fujii2017verifiable,hayashi2018self,takeuchi2019npjSerfling}. However, this problem is much more challenging because the states prepared by Bob in different runs may be correlated or even entangled with each other.

To address the non-IID problem, one approach is to use the quantum de Finetti theorem \cite{CaveFS02,christandl2007one, renner2007symmetry, PhysRevLett.114.160503}, in which permutation invariance plays a crucial role.  
That is, if we randomly permute a sequence of non-IID quantum systems and trace out some subset, then the joint reduced state on the remaining systems is approximated by a probabilistic mixture of IID states.
Based on this intuition, Refs.~\cite{PhysRevLett.109.120403, van2013quantum,PhysRevA.87.062331} showed that the problems of state tomography and entanglement verification can be solved in the non-IID scenario by randomizing the order of different measurements.

By using a similar idea based on randomized measurements, several protocols \cite{PhysRevA.96.062321,PhysRevX.8.021060,takeuchi2019npjSerfling, hayashi2015PRLverifiable,fawzi2024learning} have been proposed for verifying many important quantum states in the adversarial scenario, including graph states, hypergraph states, and ground states of local Hamiltonians.  
However, the sample costs of these protocols are too prohibitive for any practical application. 
Recently, a general defensive QSV (DQSV) protocol was developed in Refs.~\cite{zhu2019advLett, zhu2019advPRA}. 
Thanks to this protocol, any pure state can be verified efficiently in the adversarial scenario if it can be verified efficiently in the IID scenario. 
All protocols presented in Refs.~\cite{PhysRevA.96.062321, PhysRevX.8.021060,takeuchi2019npjSerfling,zhu2019advLett,zhu2019advPRA,hayashi2015PRLverifiable,fawzi2024learning} 
first randomly choose some systems of the joint state and  perform some test on each of them. If the test results satisfy certain conditions, then one can provide a fidelity certificate for the conditional reduced state on the remaining system. 
Compared with verification protocols in Refs.~\cite{PhysRevA.96.062321, PhysRevX.8.021060, takeuchi2019npjSerfling,fawzi2024learning}, the DQSV protocol in Refs.~\cite{zhu2019advLett, zhu2019advPRA} 
can achieve a much higher sample efficiency
because it relies on a direct optimization method that avoids unnecessary relaxation or approximation.

However, the original DQSV protocol \cite{zhu2019advLett, zhu2019advPRA} can  reach a meaningful conclusion only if all  $N$ tests output successes. 
In practice, quantum states are never perfect, and even states with  high fidelities can cause failures with significant probability when $N$ is large \cite{Li2023RobustQSV}. 
To construct a robust and efficient
verification protocol, therefore, it is crucial to consider
the general case in which some failures are observed among the $N$ tests. 
This generalization was achieved recently in Ref.~\cite{Li2023RobustQSV}, which proposed an efficient DQSV protocol that can tolerate state imperfections.

In the present work, 
by virtue of the robust DQSV recipe developed in Ref.~\cite{Li2023RobustQSV}, we experimentally verify the two-qubit singlet state in the adversarial scenario. 
To this end, we employ a high-speed preparation-and-measurement apparatus connected to quantum random-number generators (QRNGs), which enables random and rapid control and measurements on each system. 
Our experiments demonstrate that the singlet can be verified robustly and efficiently 
 in the adversarial scenario with local Pauli measurements, and the efficiency is comparable to the counterpart in SQSV. Moreover, DQSV can certify fidelity reliably even if SQSV fails.
 Notably, by consuming only 1000 samples, we can verify that the prepared singlet state has at least 97\% fidelity with a 95\% confidence level.  When no failures are observed among all tests, the verification precision $\epsilon$ is inversely proportional to the number $N$ of tests employed as in the IID scenario.

\renewcommand{\figurename}{FIG.}

\begin{figure}[bt]
\centering
\includegraphics[width=0.9\columnwidth]{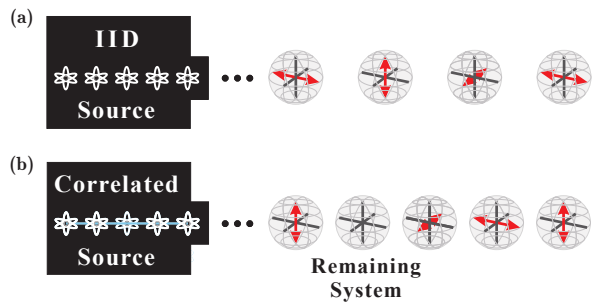}
\caption{\label{scheme}Illustration of the (a) SQSV and (b) DQSV protocols for the two-qubit singlet. (a) In SQSV, the verification strategy $\Omega$ characterized by Eq. (\ref{eq:BellStrategy}) is performed on each of the $N$ states produced, then a fidelity certificate for the unconditional reduced state on one system is constructed based on the test outcomes (under the IID assumption). Here $\Omega$ is composed of three random tests (represented by thick lines within each sphere), and the actual test chosen in each run is marked with a double arrow. (b) In DQSV, by contrast, the source may prepare a correlated or even entangled quantum state on the whole space $\mathcal{H}^{\otimes(N+1)}$. The strategy $\Omega$ is performed on each of the $N$ systems that are chosen randomly, and then a fidelity certificate for the conditional reduced state on the remaining system is constructed.}
\end{figure}

\section{Theoretical framework}
As shown in Fig.~1, both SQSV and DQSV employ a homogeneous strategy \cite{zhu2019advLett, zhu2019advPRA} based on
Pauli measurements to verify whether the prepared states are sufficiently close to the target state $|\Psi\>=(|01\>-|10\>)/\sqrt{2}\in \mathcal{H}$. The strategy is constructed as follows \cite{pallister2018PRL,zhu2019MES}. 
In each run, one performs a random test by choosing a measurement setting from $\{X \otimes X, Y \otimes Y, Z \otimes Z\}$ uniformly,  where $X, Y,$ and $ Z$ are Pauli operators,
and the test is passed if the outcome is $-1$. The resulting strategy is characterized by the verification operator  
\begin{align}\label{eq:BellStrategy}
\Omega =\frac{1}{3}\left( P_{XX}^- + P_{YY}^- +P_{ZZ}^-\right)
=|\Psi\rangle\langle\Psi|+\lambda(\openone-|\Psi\rangle\langle\Psi|),
\end{align}
where $P_{XX}^-$ is the projector onto the negative eigenspace of $X\otimes X$ (and likewise for $P_{YY}^-$ and $P_{ZZ}^-$) and $\lambda=1/3$. Denote by $\nu:=1-\lambda$ the spectral gap of $\Omega$ from the largest eigenvalue.

In addition to QSV, the homogeneous strategy in \eref{eq:BellStrategy} can also be applied to fidelity estimation \cite{zhu2019advPRA}. 
Note that the fidelity $F$ between a quantum state $\sigma$ and the target state $|\Psi\>$ can be written as 
\begin{align}\label{eq:Festimate}
F(\sigma)=\<\Psi|\sigma|\Psi\>=\frac{\tr(\Omega\sigma)-\lambda }{1-\lambda}. 
\end{align}
Hence, the true fidelity can be efficiently estimated from the
probability $\tr(\Omega\sigma)$ that $\sigma$ passes a test. In this way, we can efficiently benchmark the applicability and performance of SQSV and DQSV protocols.

\subsection{Framework of SQSV}
In the framework of SQSV, a quantum source is supposed to produce the target state $|\Psi\rangle$, but 
actually produces $N$ copies of  $\sigma$  in $N$ runs. The quality of the source can be characterized by the fidelity of the unconditional reduced state on one system, that is, the state $\sigma$, given the premise that all the copies are subject to the IID assumption. By applying the strategy $\Omega$ in each run, our goal is to verify whether  $\sigma$ is sufficiently close  to the target state.   Let $p_k(\sigma^{\otimes N})$ be the probability that we observe at most $k$ failures among the $N$ tests. 
To characterize the verification precision, for  a given  significance level $0<\delta\leq1$ and allowed number of failures
$k\in[0,N-1]$, we define the SQSV fidelity certificate as
\begin{align}\label{eq:SQSVguarteeF}
\mathcal{F}_{\lambda }^{\rm S}(k,N,\delta)&:=
\min_{\sigma} \left\{F(\sigma) \,\big|\, p_k\bigl(\sigma^{\otimes N}\bigr)\geq \delta  \right\} , 
\end{align}
where the minimization is over all quantum states $\sigma$ on $\caH$, and $F(\sigma)=\<\Psi|\sigma|\Psi\>$ is the fidelity between $\sigma$ and the target state.
Accordingly, the guaranteed infidelity is defined as $\epsilon_{\lambda }^{\,\rm S}(k,N,\delta):= 1-\mathcal{F}_{\lambda }^{\rm S}(k,N,\delta)$.
By definition, if at most $k$ failures are observed, then we can guarantee (with significance level $\delta$) that the
fidelity of the unconditional reduced
state is at least $\mathcal{F}_{\lambda }^{\rm S}(k,N,\delta)$. 
An analytical formula for $\mathcal{F}_{\lambda }^{\rm S}(k,N,\delta)$ is presented in Appendix~\ref{app:SQSV}.

\subsection{Framework of DQSV}
Next, we turn to the adversarial scenario, in which Bob may prepare a quantum state $\rho$ on the whole space $\caH^{\otimes (N+1)}$ and send it to the verifier, Alice. If Bob is honest, then he is supposed to prepare $N+1$ copies of $|\Psi\>$; while if he is malicious, then
he can mess up Alice by generating
an arbitrary correlated or even entangled state $\rho$ \cite{zhu2019advLett,zhu2019advPRA,hayashi2015PRLverifiable}. In order to obtain a reliable quantum state that is sufficiently close to $|\Psi\>$ for application, 
Alice needs to verify Bob's state by performing suitable tests on $N$ systems of $\rho$. 
If the test results satisfy certain conditions, then she concludes that the conditional reduced state on the remaining system is sufficiently close to the target state $|\Psi\>$ and can be used for application; otherwise, the state is rejected.

To achieve this verification task, our DQSV protocol runs as follows.  
After receiving the state $\rho$ prepared by Bob, Alice 
randomly chooses $N$ systems and then applies the strategy $\Omega$ constructed in \eref{eq:BellStrategy} to test each chosen system.  She accepts the reduced state $\sigma_{N+1}$ on 
the remaining system if and only if at most $k$ failures are observed among the $N$ tests. Since the $N$ systems are chosen
randomly, we can assume, without loss of generality, that $\rho$ is permutation-invariant
and the tests are applied on the first $N$ systems of $\rho$ \cite{zhu2019advLett,zhu2019advPRA}. 

The original DQSV protocol can  reach a meaningful conclusion only when  all $N$ tests are passed, that is, $k=0$ \cite{zhu2019advLett,hayashi2015PRLverifiable,zhu2019advPRA}. In practice, however, quantum states are never perfect, and even quantum states with high fidelity can  cause a few failures with significant probability when  $N$ is large \cite{Li2023RobustQSV}. To construct a robust and efficient verification protocol, therefore, it is crucial to consider the general case with $k\geq 0$.

To characterize the verification precision of DQSV, for a given  $0<\delta\leq1$ and integer 
$k\in[0,N-1]$ we define the DQSV fidelity certificate as
\begin{align}
\mathcal{F}_{\lambda }^{\rm D}(k,N,\delta)&:=
\min_{\rho} \left\{ F_k(\rho)\,|\, p_k(\rho)\geq \delta  \right\}, \label{eq:hatzeta2} 
\end{align}
where $p_k(\rho)$ is the probability that Alice observes at most $k$ failures among the $N$ tests, $F_k(\rho)=\<\Psi|\sigma_{N+1}|\Psi\>$ is the fidelity between the conditional reduced state $\sigma_{N+1}$ and the target state $|\Psi\>$, and the minimization is taken over all permutation-invariant states $\rho$ on ${\cal H}^{\otimes (N+1)}$.
Accordingly, the guaranteed infidelity reads $\epsilon_{\lambda }^{\,\rm D}(k,N,\delta):= 1-\mathcal{F}_{\lambda }^{\rm D}(k,N,\delta)$. If at most $k$ failures are observed among the $N$ tests, then Alice can ensure (with significance level~$\delta$) that the state $\sigma_{N+1}$ has fidelity at least $\mathcal{F}_{\lambda }^{\rm D}(k,N,\delta)$ with the target state $|\Psi\>$. An analytical formula for $\mathcal{F}_{\lambda }^{\rm D}(k,N,\delta)$ is presented in Appendix~\ref{app:DQSV}.

\begin{figure}[tb]
\centering
\includegraphics[width=0.95\columnwidth]{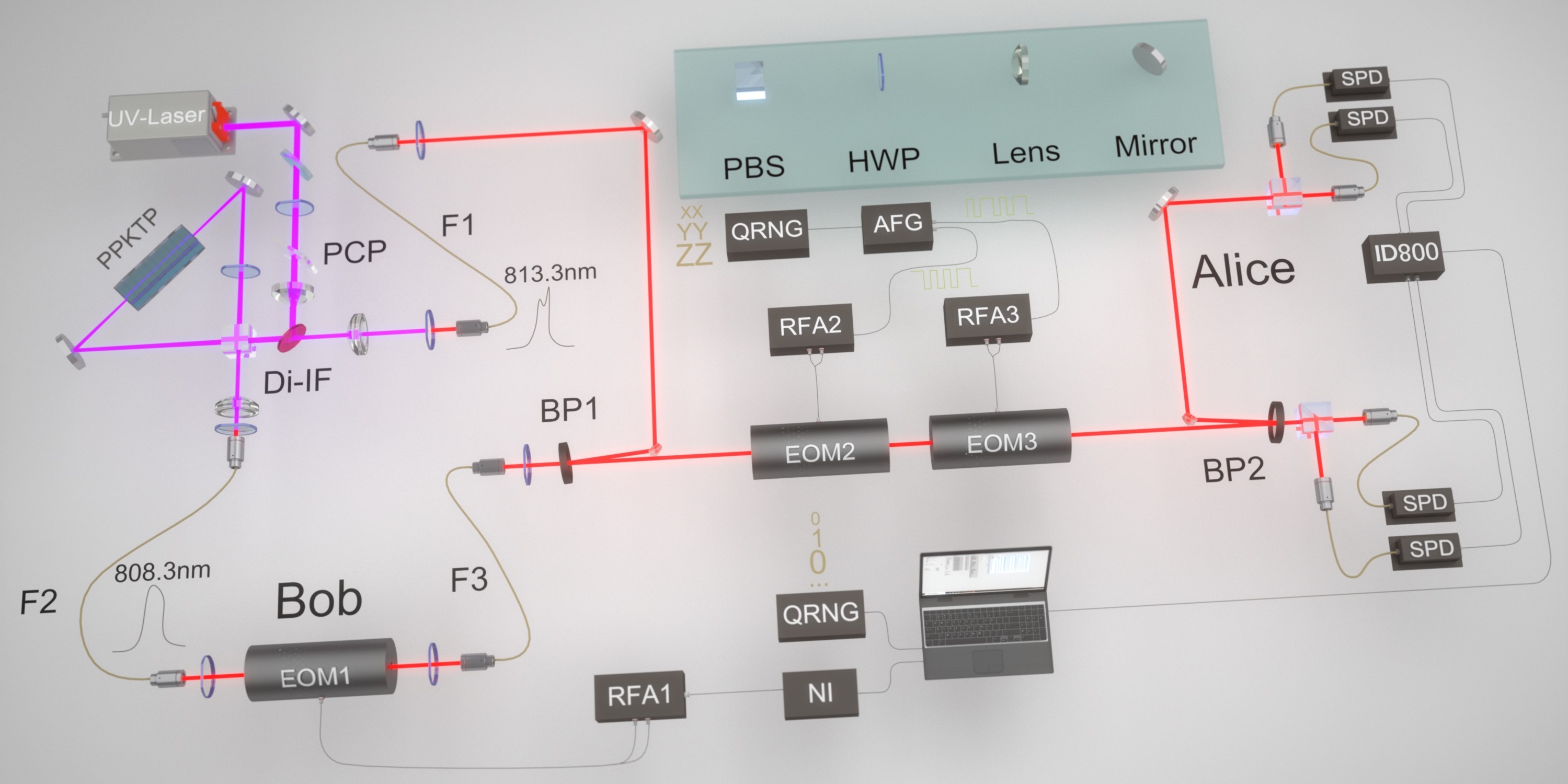}
\caption{\textbf{Setup for high-speed DQSV.} After the emission of a nondegenerate photon pair (808.3 and 813.3 nm) from an entangled-photon-pair source, Bob may undermine the entangled state by modulating the phase of the photon passing EOM1, which aims to  mislead Alice to approve the source. The driving pulses for EOM1 are generated by a digital acquisition card (NI) and controlled by a QRNG (ID Quantique), and then amplified up to a level sufficient to drive EOM1 by a radio-frequency amplifier (RFA). The nondegenerate photon pair is combined to be collinear by a bandpass filter (BP1), and then Alice can implement a random projective test on the pair with EOM2 and EOM3, both of which are driven by amplified signals from an arbitrary function generator (AFG) and controlled by another independent QRNG. The two nondegenerate photons are separated by a second bandpass filter (BP2). The final coincidence detection is executed by four single-photon detectors (SPDs) and a following ID800 to record the time stamps.}
\label{setup}
\end{figure}

\section{Experimental results}

\subsection{Experimental setup}
The experimental setup for implementing the SQSV and DQSV protocols is shown in Fig.~2. On request from Alice, Bob prepares and sends $N+1$ systems, which are all supposed to be in the singlet state, to Alice one by one, and Alice decides the local measurement on each system. To defend against Bob's attacks, Alice should execute the random measurements at a speed much higher than Bob's modulation rate to ensure that each system is measured independently and randomly. In the experiment, Bob modulates the passing photon at the speed of 1 kHz, which is much lower than Alice's switching rate at 1 MHz. The randomness of Alice's and Bob's operations is guaranteed by two QRNGs (see Appendix~\ref{app:experiment} for more details).

\begin{figure}[bt]
	\centering
\includegraphics[width=0.95\columnwidth]{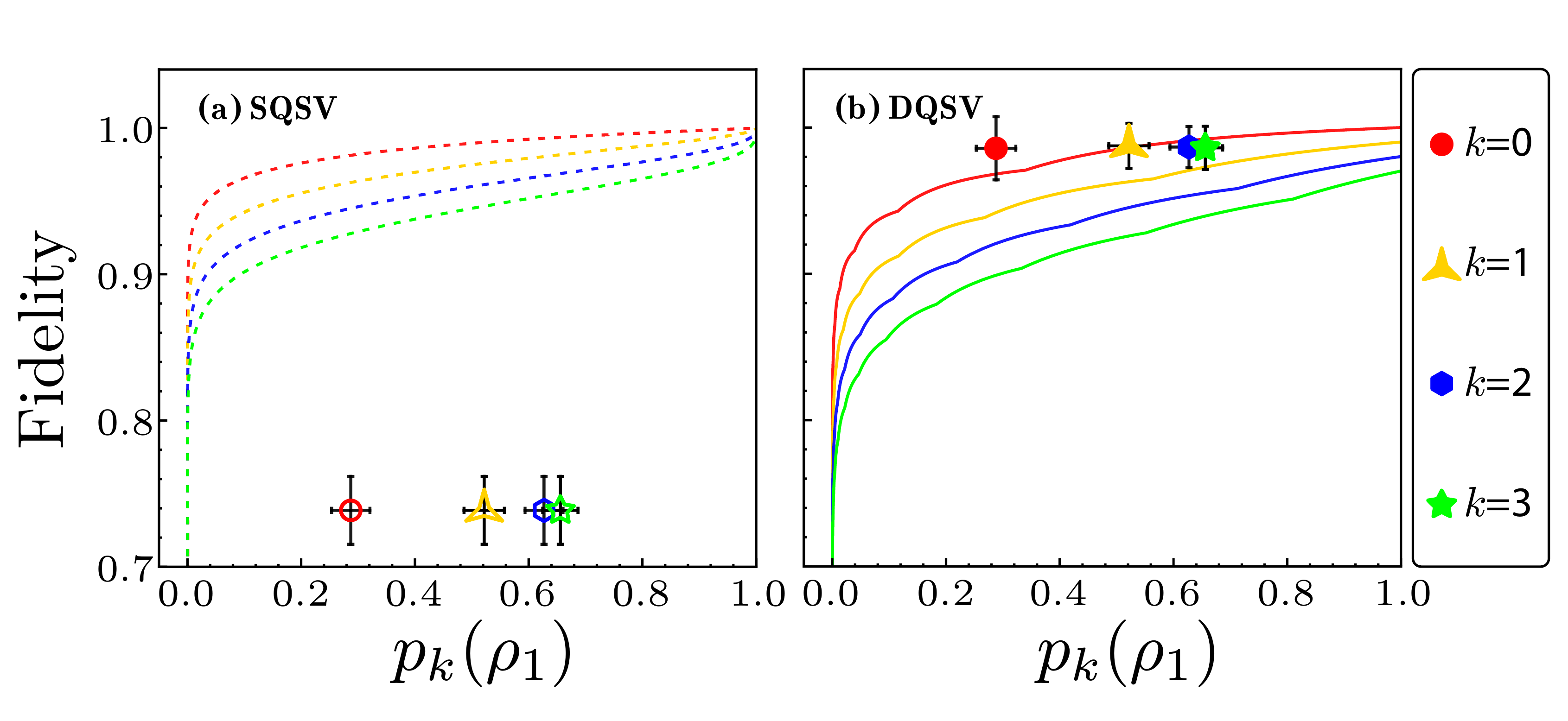}	\caption{Verification of the  correlated state $\rho_1$ in \eref{correlated} based on SQSV and DQSV protocols.  Here $k$ denotes the maximum number of failures allowed among $N=100$ tests.
(a) In SQSV, the true fidelity of  the unconditional reduced state (independent of $k$) is marked with a symbol. 
The corresponding SQSV certificate, represented by a dashed curve, is violated, which means SQSV is not applicable in this case. 
(b) In DQSV, the true fidelity of the reduced state on the remaining system conditioned on observing at most $k$ failures  is marked with a symbol. 
The corresponding DQSV certificate   represented by a solid curve offers a reliable and nearly tight lower bound. 
The error bar for each symbol (in both plots) is calculated from 40 rounds of experiments.}
\label{DQSV vs SQSV}
\end{figure}

\subsection{Verification in the non-IID scenario}
To examine whether SQSV and DQSV can provide reliable fidelity certificates in the non-IID scenario,  we apply them to verify two kinds of correlated states. 
To  benchmark the verification results,  the true fidelities are estimated using 
 \eref{eq:Festimate}.
First, we consider the case in which Bob prepares a correlated  state of the form 
\begin{align}
\label{correlated}
\rho_1=\frac{2}{3}\left( |\Psi\> \<\Psi|\right) ^{\otimes (N+1)}+\frac{1}{3}\left( \frac{\openone}{4}\right) ^{\otimes (N+1)}
\end{align}
with $N=100$.
In other words, Bob prepares $N+1$ copies of the target state $|\Psi\>$ with probability $2/3$ and 
$N+1$ copies of the completely mixed state with probability $1/3$. In the experiment,  when preparing $|\Psi\>$, the inevitable noise leads to a nonunity fidelity of about 98\%. 
In SQSV, we use the strategy $\Omega$ in \eref{eq:BellStrategy} to test each of the
first $N$ systems of $\rho_1$. 
In DQSV, we use the same strategy to test each of the $N$ systems that are chosen randomly. 
In both cases, the verification procedures are
repeated 200 times to determine the probability $p_k(\rho_1)$ of observing at most $k$ failures among the $N$ tests. 
The guaranteed fidelities provided by SQSV and DQSV are determined by Eqs.~\eqref{eq:SQSVguarteeF} and \eqref{eq:hatzeta2} with $\delta=p_k(\rho_1)$ and are shown as  curves in Figs.~3(a) and 3(b), respectively.

\begin{figure}[t]
\centering
\includegraphics[width=0.95\columnwidth]{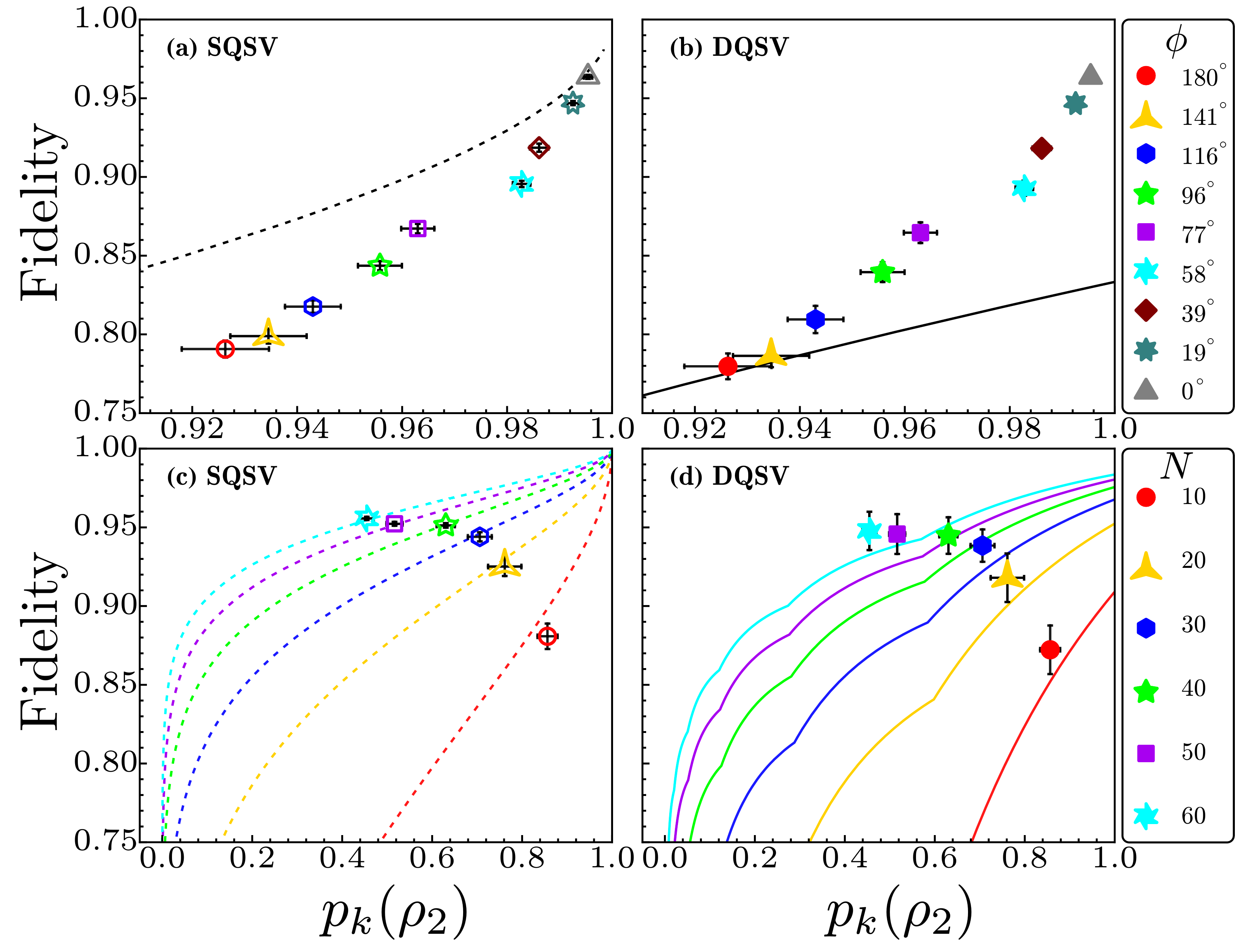}
\caption{Verification of  correlated states $\rho_2$ that have the form in \eref{correlated2} based on SQSV and DQSV protocols.  
In panels (a) and (c) [panels (b) and (d)], the true fidelities of the unconditional (conditional) reduced states are marked with symbols, while the corresponding SQSV (DQSV) certificates are represented by curves.
Again, the SQSV fidelity certificates are reliable for neither the conditional nor unconditional reduced states, and can be significantly violated. 
In contrast, the DQSV fidelity certificates (for the conditional reduced states) are always reliable and usually quite tight. Here  $N=5$ for panels (a) and (b), while $\phi=\pi$ for panels (c) and (d). For all four plots, $k=1$ and the error bar for each symbol is calculated from 40 rounds of experiments.
}
\label{defensive}
\end{figure}

As a benchmark for SQSV, the average fidelity of the first $N$ systems is estimated via \eref{eq:Festimate}, which characterizes the true fidelity of the unconditional reduced state of $\rho_1$ and is shown as a symbol in Fig.~3(a). 
It is independent of $k$ as expected, but is significantly lower than the theoretical lower bound. This fact means SQSV cannot deal with a correlated source, let alone the general adversarial scenario. By contrast, as a benchmark for our DQSV, after performing the $N$ tests, the fidelity of the reduced state on the remaining system conditioned on observing at most $k$ failures is estimated via \eref{eq:Festimate} and shown as a symbol in Fig.~3(b). This fidelity tends to decrease with $k$, but the variation is almost invisible for small $k$ because such events are dominated by a noisy version of the  term $( |\Psi\> \<\Psi|) ^{\otimes (N+1)}$ in \eref{correlated}. 
Importantly, this fidelity is  slightly higher than the corresponding fidelity certificate, which means  DQSV can provide a reliable and nearly tight certificate.

To further corroborate the necessity of DQSV in the adversarial scenario, next, we consider the case in which Bob prepares a 
 permutation-invariant quantum state of the form
\begin{equation} \label{correlated2}
\rho_2=\frac{1}{N+1}\sum_l\mathscr{P}_l \left\lbrace (|\Psi\rangle\langle\Psi|)^{\otimes N}\otimes |\Psi(\phi)\rangle\langle\Psi(\phi)| \right\rbrace ,
\end{equation}
where $\sum_l\mathscr{P}_l \{\cdot\}$ denotes the sum over all distinct permutations. Here $|\Psi(\phi)\>=(|01\>-e^{i\phi}|10\>)/\sqrt{2}$ is generated by applying an extra phase $\phi$ between the horizontal polarization and vertical polarization and is randomly placed into the sequence of $N+1$ systems.
During the experiment, we apply both SQSV and DQSV protocols to construct  fidelity certificates for various choices of $N$ and $\phi$. To reduce the statistical fluctuation in determining  $p_{k}(\rho_2)$, the verification procedure in each case is repeated until the event of observing at most $k$ failures occurs 1000 times.  
As shown in Fig.~4, the certificates provided by SQSV are often higher than the measured true fidelities, not only for the unconditional reduced states but also for the conditional reduced states on the remaining system, indicating the inability of SQSV to handle correlated sources.
In contrast, the fidelity certificates provided by DQSV are consistently slightly lower than the measured true fidelities of the conditional reduced states. 
This result further demonstrates that DQSV can overcome the limitation of SQSV and can provide reliable and informative fidelity certificates in the adversarial scenario.

\begin{figure}[t]
\centering
\includegraphics[width=0.95\columnwidth]{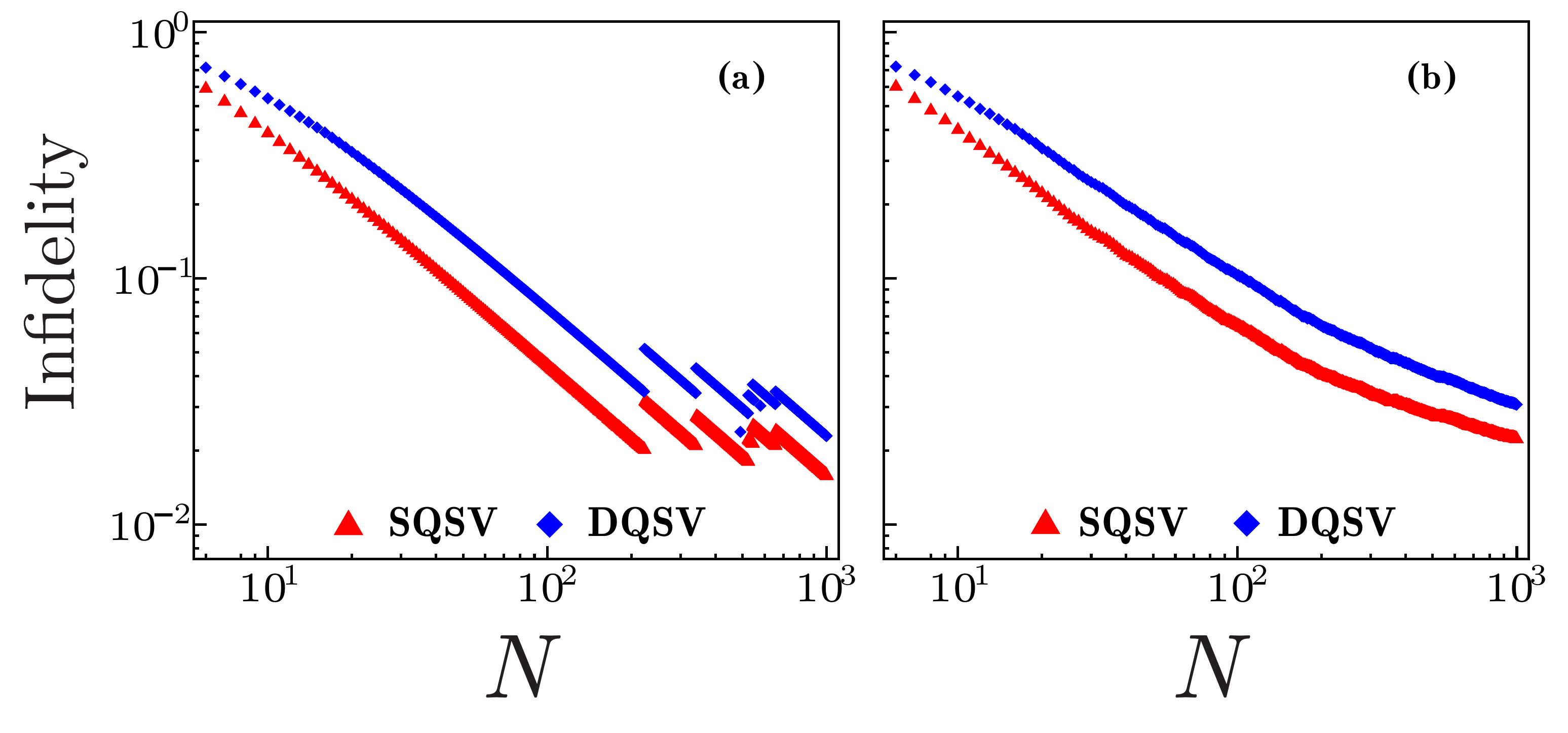}
\caption{\label{efficiency}Verification precisions achieved by SQSV and DQSV when Bob is honest. Here the significance level $\delta$ is set to be 0.05, and the guaranteed infidelity $\epsilon$ is  plotted versus $N$ (in log-log scale). The blue squares (red triangles)  represent the infidelity certificates $\epsilon_{\lambda}^{\,\rm D}(k,N,\delta)$ $\left[ \epsilon_{\lambda}^{\,\rm S}(k,N,\delta)\right]$ guaranteed by DQSV (SQSV). Panel (a) shows the results for a single round, while panel (b) shows the results averaged over  80 rounds. Both $\epsilon_{\lambda}^{\,\rm D}$ and 
$\epsilon_{\lambda}^{\,\rm S}$ 
scale as $N^{-1}$ within the first 100 tests. }
\end{figure}

\subsection{Verification efficiency}
Finally, we  demonstrate the efficiency of DQSV in comparison with SQSV when Bob is honest. In this case, Bob tries to prepare the ideal singlet state $|\Psi\>$ (the target state) in each run, but the actual state prepared may not be perfect due to inevitable noise
 and other imperfections. The performance of the  SQSV and DQSV protocols are characterized by the scalings of the guaranteed infidelities $\epsilon_{\lambda}^{\,\rm S}(k,N,\delta)$ and
$\epsilon_{\lambda}^{\,\rm D}(k,N,\delta)$  with respect to the number $N$ of tests.  Here $k$ is the number of failures actually observed among the $N$ tests, while   $\epsilon_{\lambda}^{\,\rm S}(k,N,\delta)$  and $\epsilon_{\lambda}^{\,\rm D}(k,N,\delta)$ are determined by Proposition~\ref{prop:IIDguar} in Appendix~\ref{app:SQSV} and Proposition~\ref{prop:AdvFidelity} in Appendix~\ref{app:DQSV}, respectively.
In the experiment, the significance level $\delta$ is set to be 0.05, and 1000 tests in total are performed in each round.  The experimental  results for a single round are shown in Fig.~\ref{efficiency}~(a), while the averages over 80 rounds are shown in Fig.~\ref{efficiency}~(b). These results indicate that both $\epsilon_{\lambda}^{\,\rm D}$ and 
$\epsilon_{\lambda}^{\,\rm S}$ scale as $N^{-1}$ within the first 100 tests, and eventually they are determined by the failure rate  $k/N$ and the parameter $\lambda$. 
Notably,  $\epsilon_{\lambda}^{\,\rm D}$ drops below 0.03 within 1000 tests 
for both the single-round and average results, which is comparable to the counterpart in SQSV. In this way, our experiments demonstrate that DQSV can achieve a similar high efficiency as SQSV although its underlying assumption is much weaker. 
Incidentally, previous verification protocols applicable to the adversarial scenario \cite{fawzi2024learning,PhysRevX.8.021060,PhysRevA.96.062321,hayashi2015PRLverifiable,takeuchi2019npjSerfling} are too resource-intensive to be implemented in experiments.

\section{Summary}
Using a photonic platform, we experimentally realized a robust DQSV protocol for the singlet state in the adversarial scenario using local Pauli measurements. The experimental results show that our DQSV protocol can successfully defend against various attacks from a correlated or even malicious source, while the SQSV protocol may fail. Moreover, our DQSV protocol is robust to imperfections in state preparation and can achieve a high efficiency that is comparable to that of the  SQSV protocol.  Our work demonstrates that the DQSV protocol is a powerful tool for verifying quantum states in the adversarial scenario, 
which is free of the IID assumption. It may play important roles in various quantum information processing tasks that require high-fidelity state preparation, but suffer from  correlated noise or untrustworthy sources.

\section*{Acknowledgments}
This work is supported by the Innovation Program for Quantum Science and Technology (Grant No. 2021ZD0301200), the National Natural Science Foundation of China (Grants No. 12350006 and No. 62205326), the USTC Research Funds of the Double First-Class Initiative (Grant No. YD2030002026), and the Natural Science Foundation of Anhui Province of China (Grant No. 2408085Y002). The work at Fudan University is supported by Shanghai Science and Technology Innovation Action Plan (Grant No.~24LZ1400200), Shanghai
Municipal Science and Technology Major Project
(Grant No.~2019SHZDZX01), the Innovation Program for Quantum Science and Technology (Grant No. 2024ZD0300101), and the National Key Research and Development Program
of China (Grant No.~2022YFA1404204).

\begin{appendix}
\numberwithin{proposition}{section}
\numberwithin{lemma}{section}
\renewcommand{\theproposition}{\thesection\arabic{proposition}}
\renewcommand{\thelemma}{\thesection\arabic{lemma}}

\section{Fidelity certificate in SQSV}\label{app:SQSV}
As mentioned in the main text, in SQSV, the verifier applies the strategy $\Omega$ on $N$ copies of the  state $\sigma$ produced.
Let $\epsilon_\sigma:=1-\<\Psi|\sigma|\Psi\>$ be the infidelity between $\sigma$ and the target state $|\Psi\rangle \in \mathcal{H}$. 
Then the maximum probability that $\sigma$ can pass a test on average is given by
\cite{pallister2018PRL,zhu2019advLett}  \begin{equation}\label{eq:NAPrPass1test}
\max_{\<\Psi|\tau|\Psi\> \leq  1-\epsilon_\sigma} \tr(\Omega \tau) =  1- \nu \epsilon_\sigma,
\end{equation}
where the maximization runs over all quantum states $\tau$ on $\caH$ that satisfy $\<\Psi|\tau|\Psi\> \leq  1-\epsilon_\sigma$, and $\nu=1-\lambda$ denotes the spectral gap of $\Omega$ from the largest eigenvalue \cite{pallister2018PRL,zhu2019advLett}.  The maximum is attained when $\tau$ is supported in the subspace associated with the largest and second largest eigenvalues of $\Omega$.

For $0\leq p \leq 1$ and integers $z,k\geq 0$, define 
\begin{align}\label{eq:binomCFD}
B_{z,k}(p):= \sum_{j=0}^k {z \choose j} p^j (1-p)^{z-j}.
\end{align}
Here it is understood that $x^0=1$ even if $x=0$.
Then the probability that we observe at most $k$ failures  among the $N$ tests ($0\leq k\leq N-1$) can be expressed as  
\begin{align}\label{eq:pkleq}
p_k\bigl(\sigma^{\otimes N}\bigr)
&=\sum_{j=0}^k {N \choose j} [1-\tr(\Omega \sigma)]^j [\tr(\Omega \sigma)]^{N-j}
\nonumber\\
&=B_{N,k}[1-\tr(\Omega \sigma)]
\leq B_{N,k}(\nu\epsilon_\sigma),
\end{align}
where the inequality follows from \eref{eq:NAPrPass1test} above and \lref{lem:Bzkmono} below, and can be saturated when $\sigma$ is supported in the subspace associated with the largest and second largest eigenvalues of $\Omega$.

As long as the number of failures among the $N$ tests is not larger than $k$, one can guarantee (with significance level $\delta$) that the fidelity of $\sigma$ is at least $\mathcal{F}_{\lambda }^{\rm S}(k,N,\delta)$,  where 
\begin{align}
\mathcal{F}_{\lambda }^{\rm S}(k,N,\delta)&:=
\min_{\sigma} \left\{\<\Psi|\sigma|\Psi\> \,\big|\, p_k\bigl(\sigma^{\otimes N}\bigr)\geq \delta  \right\} 
\end{align}
is the guaranteed fidelity of SQSV [cf.~\eref{eq:SQSVguarteeF} in the main text]. An analytical formula for $\mathcal{F}_{\lambda }^{\rm S}(k,N,\delta)$ is presented in the following proposition.

\begin{proposition}\label{prop:IIDguar}
Suppose $0\leq \lambda<1$, $0<\delta\leq1$, and 
integers $k \geq 0$ and $N\geq k+1$. Then we have
\begin{align}
\mathcal{F}_{\lambda}^{\rm S}(k,N,\delta)
&= 1- \nu^{-1} J(N,k,\delta), 
\end{align}
where $J(N,k,\delta)$ is the unique solution of $x$ to the equation
\begin{align}
B_{N,k}(x)=\delta,\quad 0\leq x\leq 1. 
\end{align}
\end{proposition}

The function $J(N,k,\delta)$ above is well-defined because 
$B_{N,k}(x)$ is continuous and strictly decreasing in $x$ for $0\leq x\leq1$,   
and $B_{N,k}(0)=1\geq\delta>0=B_{N,k}(1)$.

\begin{proof}[Proof of Proposition~\ref{prop:IIDguar}]
By definition we have
\begin{align}\label{eq:SQSVFLBproofEq1}
\mathcal{F}_{\lambda}^{\rm S}(k,N,\delta)
&= 1-\max_{\sigma} \left\{ \epsilon_\sigma \big|\, p_k\bigl(\sigma^{\otimes N}\bigr)\geq \delta  \right\}. 
\end{align}
In addition, we can derive the following result:
\begin{align}\label{eq:bound1}
\max_{\sigma} \left\{ \epsilon_\sigma \big|\, p_k\bigl(\sigma^{\otimes N}\bigr)\geq \delta  \right\}
&\leq \max_{\sigma} \left\{ \epsilon_\sigma \big| B_{N,k}(\nu\epsilon_\sigma)\geq \delta  \right\}
\nonumber\\
&= \nu^{-1} J(N,k,\delta),  
\end{align}
where the inequality follows from \eref{eq:pkleq}, 
and the equality holds because $B_{N,k}(x)$ is continuous and strictly decreasing in $x$ for $0\leq x\leq1$ according to \lref{lem:Bzkmono} below.

On the other hand, if the state $\sigma$ is supported in the subspace associated with the largest and second largest eigenvalues of $\Omega$ and $\epsilon_\sigma=\nu^{-1} J(N,k,\delta)$, then by \eref{eq:pkleq} we have 
\begin{align} 
p_k\bigl(\sigma^{\otimes N}\bigr)
=B_{N,k}(\nu\epsilon_\sigma)
=B_{N,k}(J(N,k,\delta))
=\delta, 
\end{align}
which implies that 
\begin{align}\label{eq:bound2}
\max_{\sigma} \left\{ \epsilon_\sigma \big|\, p_k\bigl(\sigma^{\otimes N}\bigr)\geq \delta  \right\}
\geq \nu^{-1} J(N,k,\delta). 
\end{align}
Equations~\eqref{eq:SQSVFLBproofEq1}, \eqref{eq:bound1}, and \eqref{eq:bound2} above together confirm Proposition~\ref{prop:IIDguar}. 
\end{proof}

\begin{lemma}[Lemma 3.2, \cite{zhu2022nearly}]\label{lem:Bzkmono}
	Suppose integers $k,z \geq 0$, $0 \leq k \leq z$, and $0<p<1$. Then $B_{z, k}(p)$ is strictly increasing in $k$, strictly decreasing in $z$, and
	nonincreasing in $p$.
	In addition, $B_{z, k}(p)$ is strictly decreasing in $p$ when $k<z$.
\end{lemma}

\section{Fidelity certificate in DQSV}\label{app:DQSV}
As mentioned in the main text, in DQSV, Alice applies  the homogeneous strategy $\Omega$ to test $N$ systems of the received state $\rho$, and accepts the reduced state on the remaining system  if at most $k$ failures are observed among the $N$ tests.
Since $\rho$ can be assumed to be permutation-invariant without loss of generality, the probability that Alice accepts $\rho$ reads 	\cite{Li2023RobustQSV}
\begin{align}\label{eq:pnmkrho}
p_k(\rho)=
\sum_{i=0}^k {N \choose i}
\tr \! \big(\big[\Omega^{\otimes (N-i)}\otimes \bar{\Omega}^{\otimes i} \otimes \openone\big] \rho\big),
\end{align}
where $\bar{\Omega}:=\openone-\Omega$.
Let $\sigma_{N+1}$ be the reduced state on the remaining system conditioned on observing at most $k$ failures. If $p_k(\rho)>0$, then the fidelity between $\sigma_{N+1}$ and the target state $|\Psi\>$ reads \cite{Li2023RobustQSV}
\begin{align}
F_k(\rho) = \<\Psi|\sigma_{N+1}|\Psi\>=\frac{f_k(\rho)}{p_k(\rho)},
\end{align}
where	
\begin{align}\label{eq:fnmkrho}
\begin{split}
    f_k(\rho) =  
    \sum_{i=0}^k {N \choose i} \tr \! \big(\big[
    \Omega^{\otimes (N-i)} \otimes \bar{\Omega}^{\otimes i} 
    \otimes |\Psi\>\<\Psi|\big] \rho \big).
\end{split}
\end{align}

As long as the number of failures among the $N$ tests is not larger than $k$, Alice can guarantee (with significance level $\delta$) that the fidelity of the reduced state on the remaining system is at least $\mathcal{F}_{\lambda }^{\rm D}(k,N,\delta)$,  where 
\begin{align}
\mathcal{F}_{\lambda }^{\rm D}(k,N,\delta)&:=
\min_{\rho} \left\{F_k(\rho) \,|\, p_k(\rho)\geq \delta  \right\}  
\end{align}
is the guaranteed fidelity of DQSV [cf.~\eref{eq:hatzeta2}  in the main text]. 
In the following, we provide an analytical formula for $\mathcal{F}_\lambda^{\rm D}(k,N,\delta)$ by virtue of the results in Ref.~\cite{Li2023RobustQSV}.

For integer $0\leq z\leq N+1$, define
\begin{widetext}
\begin{align}
h_z(k,N,\lambda) &:=
\begin{cases}
	1 & 0\leq z \le k, \\[0.5ex]
			\frac{(N-z+1) B_{z,k}(\nu) +z B_{z-1,k}(\nu) }{N+1} & k+1 \leq z \leq N+1,
\end{cases}
		\label{eq:hzHomo} \\
		g_z(k,N,\lambda)&:=
		\begin{cases}
			\frac{N-z+1}{N+1} &  0\leq z \le k, \\[1ex]
			\frac{(N-z+1)B_{z,k}(\nu)}{N+1}   &  k+1 \leq z \leq N+1. 
		\end{cases}
		\label{eq:gzHomo}
	\end{align}
\end{widetext}	
	
	\begin{lemma}[Lemma S1, \cite{Li2023RobustQSV}]\label{lem:gzhzMono} 
		Suppose $0<\lambda<1$, integers $k,z,N \geq 0$, and $N \geq k+1$.
		Then 
		$h_z(k,N,\lambda)$ is strictly decreasing in $z$ for $k\leq z\leq N+1$, and
		$g_z(k,N,\lambda)$ is strictly decreasing in $z$ for $0\leq z\leq N+1$.
	\end{lemma}
	
	By this lemma, we have $B_{N,k}(\nu)=h_{N+1}(k,N,\lambda)\leq h_z(k,N,\lambda)\leq 1$ for $0\leq z\leq N+1$.
	Hence, for $B_{N,k}(\nu)<\delta\leq 1$, we can define $\hat{z}$ as the largest integer $z$ such that
	$h_z(k,N,\lambda)\geq\delta$.
	For $z\in\{k,k+1,\dots,N\}$, define
	\begin{align}
	\kappa_z(k,N,\delta,\lambda)
		:=\;& \frac{\delta-h_{z+1}(k,N,\lambda)}{h_z(k,N,\lambda)-h_{z+1}(k,N,\lambda)}, \label{eq:kappaz} \\
		\tilde{\zeta}_\lambda(k,N,\delta,z)
		:=\;& [1-\kappa_z(k,N,\delta,\lambda)]g_{z+1}(k,N,\lambda) \nonumber\\
		    &+\kappa_z(k,N,\delta,\lambda)g_z(k,N,\lambda).  \label{eq:tildezeta1}
	\end{align}
The exact value of $\mathcal{F}_\lambda^{\rm D}(k,N,\delta)$ is determined by Proposition~\ref{prop:AdvFidelity} below, which follows from Theorem S1 in Ref.~\cite{Li2023RobustQSV}.
	
\begin{proposition}\label{prop:AdvFidelity}
Suppose $0<\lambda<1$, $0<\delta\leq1$, and integers $k \geq 0$ and $N\geq k+1$. Then we have
\begin{align}
&\mathcal{F}_\lambda^{\rm D}(k,N,\delta)
=
\begin{cases}
0                             & \delta \leq B_{N,k}(\nu), \\
\tilde{\zeta}_\lambda(k,N,\delta,\hat{z})/\delta  \  & \delta > B_{N,k}(\nu).
\end{cases}
\label{eq:True-plot}
\end{align}
\end{proposition}

\section{Experimental details}\label{app:experiment}
To prepare the singlet state $|\Psi\>$, the continuous ultraviolet beam from a semiconductor laser (Toptica TopMode) with an average power reduced to about 10 mW and wavelength centered at 405.4 nm is first prepared in the polarization state $\frac{1}{\sqrt{2}}(|H\rangle+|V\rangle)$ by virtue of a half-wave plate and then focused on a 20-mm-long periodically poled potassium titanyl phosphate (PPKTP) crystal by a lens (see Fig.~1). The clockwise and counterclockwise pumping beams are adjusted to maximally overlap to resist decoherence from thermal and mechanical fluctuations. The PPKTP crystal is placed on a thermoelectric cooler and the whole Sagnac interferometer is covered in an acrylic box, which is temperature-stabilized within 0.001~\textcelsius. The temperature of this nonlinear crystal is set at about 44.4~\textcelsius \ to produce photon pairs with nondegenerate wavelengths centered at 808.3 nm (signal) and 813.3 nm (idle), respectively. A phase compensation plate (PCP) is tilted to yield the singlet state $|\Psi\rangle$.

A free-space EOM1 with maximum 100-MHz bandwidth is used by Bob to prepare the correlated states $\rho_1$ and $\rho_2$ in Eqs.~\eqref{correlated} and \eqref{correlated2}. One QRNG generates random numbers deciding which system in $\rho_1$ or $\rho_2$ is modulated, and then a data acquisition card (DAQ, NI PCI) programmed by Bob sends  high-speed synchronous 5-V transistor-transistor logic signals to the radio-frequency amplifier (RFA1) to drive EOM1.

To apply the verification strategy characterized by the operator $\Omega =\frac{1}{3}\left( P_{XX}^- + P_{YY}^- +P_{ZZ}^-\right) $, Alice randomly chooses one of the three Pauli operators $X, Y$, or $Z$, and performs the same Pauli measurement on the two photons. A 808-nm bandpass interference filter combines the signal photon (808.3 nm) and the reflected idle photon (813.3 nm) to be collinear, and hence Alice can apply the strategy  $\Omega$ using a set of EOM2 and EOM3. The wavelengths of signal and idle photons are only slightly nondegenerate so as to keep the birefringence effect imposed by the EOMs almost identical, and also the projective measurements on the two photons.

 The other QRNG is used to switch the three measurement settings employed in the strategy $\Omega$ by controlling the voltages of EOM2 and EOM3.  The fast axes of EOM2 and EOM3 are set to deviate from the horizontal position by $0^\circ$ and $22.5^\circ$, respectively. Both EOMs have two possible voltage settings: in one setting both EOM2 and EOM3 act as the identity operation, while in the other setting  EOM2 acts as a quarter-wave plate, and EOM3 acts as a half-wave plate. Combined with a polarization beam splitter (PBS), the setup can realize the verification strategy $\Omega$ consisting of the three measurements $XX$, $YY$, and $ZZ$ at high speed.

In the experiment, Alice's and Bob's operations are determined by two different modules of the Labview program on the same computer.   After receiving each system, Alice categorizes it into multichannel time stamps. If the coincidence results from one photon reflected and the other photon transmitted at the PBS, then this system is recorded as passing the test. 
\end{appendix}

\bibliography{cites}	

\end{document}